\documentclass[11pt,reqno]{amsart}
\usepackage{amsfonts,amsmath,amssymb}
\newtheorem{thm}{Theorem}[section]
\newtheorem{proposition}[thm]{Proposition}
\newtheorem{corollary}[thm]{Corollary}
\newtheorem{lemma}[thm]{Lemma}

\usepackage{color}
\usepackage{graphicx}
\usepackage{epstopdf}

\title[Jack--Laurent symmetric functions for special values of parameters]{Jack--Laurent symmetric functions for special  values of parameters}
\author{ A.N. Sergeev}\address{Department of Mathematics, Saratov State University, Astrakhanskaya 83, Saratov 410012, Russia and National Research University Higher School of Economics, Laboratory of Mathematical Physics,
20 Myasnitskaya Ulitsa, Moscow 101000, Russia}
\email{SergeevAN@info.sgu.ru}
\author{A.P. Veselov}
\address{Department of Mathematical Sciences,
Loughborough University, Loughborough LE11 3TU, UK  and Moscow State University, Moscow 119899, Russia}
\email{A.P.Veselov@lboro.ac.uk}

\begin{document}

\maketitle

\begin{abstract}  We consider the Jack--Laurent symmetric functions for special values of parameters $p_0=n+k^{-1}m,$ where $k$ is not rational and $m$ and $n$ are natural numbers.  In general, the coefficients of such functions may have poles at these values of $p_0.$ The action of the corresponding algebra of quantum Calogero-Moser integrals $\mathcal D(k, p_0)$ on the space of Laurent symmetric functions defines the decomposition into generalised eigenspaces. We construct a basis in each generalised eigenspace as certain linear combinations of the Jack--Laurent symmetric functions, which are regular at $p_0=n+k^{-1}m,$ and describe the action of $\mathcal D(k, p_0)$ in these eigenspaces. 
\end{abstract}

\section{Introduction}

The Jack symmetric functions $P_{\lambda}^{(k)}$ can be considered as one-parameter generalisation of Schur symmetric functions \cite{KS,Ma} and play an important role in many areas of mathematics and theoretical physics. They can be also defined as the eigenfunctions of an infinite-dimensional version of the Calogero-Moser-Sutherland (CMS) operators \cite{CMS}.

In paper \cite{SV6} we introduced and studied a Laurent version of Jack symmetric functions - {\it Jack--Laurent symmetric functions} $P^{(k,p_0)}_{\alpha}$ as certain elements of $\Lambda^{\pm}$  labelled by bipartitions $\alpha=(\lambda,\mu)$, which are pairs of the usual partitions $\lambda$ and $\mu.$ Here $\Lambda^{\pm}$ is freely generated by $p_a$ with $a \in \mathbb Z \setminus \{0\}$ being both positive and negative.
The variable $p_0$ plays a special role and is considered as an additional parameter.
The usual Jack symmetric functions $P_{\lambda}^{(k)}$  are particular cases of $P^{(k,p_0)}_{\alpha}$ corresponding to empty second partition $\mu.$  The simplest example of Jack--Laurent symmetric function corresponding to two one-box Young diagrams is given by
 $$P^{(k,p_0)}_{1,1}= p_1 p_{-1} - \frac{p_0}{1+k-kp_0}.$$
 We proved the existence of $P^{(k,p_0)}_{\alpha}$ for all $k \notin \mathbb Q$ and $p_0\neq n+k^{-1}m, \, m,n \in \mathbb Z_{>0}$ (see Theorem 4.1 in \cite{SV6}). The coefficients of  $P^{(k,p_0)}_{\alpha}$ as functions of $p_0$ are rational and may have poles at $p_0=n+k^{-1}m$ with natural $m,n,$
so the corresponding Jack--Laurent symmetric function may not exist (as one can see in the example above).
This is related to the fact that the spectrum of the algebra of the corresponding quantum CMS integrals $\mathcal D(k, p_0)$ is not simple, which leads to the decomposition of $\Lambda^{\pm}$ into generalised eigenspaces. 

In this paper we fix a non-rational value of $k$ and study the analytic properties of Jack--Laurent symmetric functions as functions of $p_0$ at the special values $p_0=n+k^{-1}m$. 
The main result is the construction of a basis in each generalised eigenspace of $\mathcal D(k, p_0)$ as certain linear combinations of the Jack--Laurent symmetric functions, which are regular at $p_0=n+k^{-1}m$. 

The structure of the paper is as follows. In the next section we introduce the equivalence relation on the set of bipartitions induced by the action of the algebra $\mathcal D(k, p_0)$ and study it in detail.
In particular, we show that each equivalence class $E$ consists of $2^r$ elements, which can be explicitly described in terms of geometry of the corresponding Young diagrams (see Fig. 1 below).

In the third section we construct the linear combinations of Jack--Laurent symmetric functions
$$
Q^{(k,p_0)}_{\alpha}=\sum_{\beta\in E,\,\, \beta\subset \alpha}a_{\beta\alpha}(k,p_0)P^{(k,p_0)}_{\beta},
$$
which are regular at $p_0=n+k^{-1}m$ and give a basis in the corresponding generalised eigenspace. Here $E$ is the equivalence class of bipartition $\alpha$ and $a_{\beta\alpha}(k,p_0)$ are some rational functions of $p_0$ with poles at $p_0=n+k^{-1}m$ of known order (see Theorem \ref{regul} below). As a corollary we describe the order of the pole of $P^{(k,p_0)}_{\alpha}$ at $p_0=n+k^{-1}m$ in terms of the geometry of the corresponding bipartition $\alpha.$ We are using the technique similar to the translation functors in the representation theory \cite{BG,Z} and based on the Pieri formula for Jack--Laurent symmetric functions derived in \cite{SV6}.

In the last section we describe the action of the algebra $\mathcal D(k, p_0)$ with $p_0=n+k^{-1}m$ in each generalised eigenspace $V_E$. More precisely, we show that provided $k$ is non-algebraic the image of $\mathcal D(k, p_0)$ in $End \, V_E$ is isomorphic to the tensor product of $r$ copies of dual numbers $\mathfrak A_r=\mathbb C[\varepsilon]^{\otimes r}$, $\varepsilon^2=0$ and the corresponding action of $\mathfrak A_r$ in $V_E$ is the regular representation of $\mathfrak A_r$.

\section{Equivalence  relation}

We start with the following result from our paper \cite{SV6} about the quantum CMS integrals at infinity.

Let us assume at the beginning that $k$ is not rational and $p_0\neq n+k^{-1}m, \, m,n \in \mathbb Z_{>0}$  and consider the corresponding Jack-Laurent symmetric function $P^{(k,p_0)}_{\alpha}$ indexed by bipartition $\alpha =(\lambda,\mu)$ (see \cite{SV6} for the precise definition).
We will use the standard representation of the partitions as Young diagrams \cite{Ma}.

\begin{thm}\cite{SV6}\label{Bern} There exist quantum CMS integrals $\mathcal B^{(r)}: \Lambda^{\pm} \rightarrow \Lambda^{\pm}$ polynomially depending on $p_0$ such that
\begin{equation}\label{bil}
\mathcal B^{(r)} P^{(k,p_0)}_{\alpha}=b_r(\alpha, k, p_0) P^{(k,p_0)}_{\alpha},
\end{equation}
where
\begin{equation}\label{bal}
b_r(\alpha, k, p_0)=\left(\sum_{x\in \lambda} c(x,0)^{r-1}
+(-1)^r \sum_{x\in \mu} c(x,1+k-kp_0)^{r-1}\right)
\end{equation}
and the content $c(x,a)$ of the box $x=(ij)$ is defined by $$c(x,a)=(j-1)+k(i-1)+a.$$
\end{thm}

The algebra of CMS integrals $\mathcal D(k, p_0)$ is generated by these operators.

Let us introduce the following equivalence relation $\mathcal E$ on bipartitions, depending on parameters $k, p_0.$ We say that $\alpha=(\lambda,\mu)$ is {\it $\mathcal E$-equivalent} to $\tilde\alpha=(\tilde\lambda,\tilde\mu)$ if and only if for all $r\ge1$ we have
$$b_r(\alpha, k, p_0) =b_r(\tilde\alpha, k, p_0),$$ or, more explicitly, 
\begin{equation}\label{eqv}
\sum_{x\in\lambda} c(x,0)^{r-1}+(-1)^{r}\sum_{y\in\mu} c(y,1+k-kp_0)^{r-1}
\end{equation}
$$ =\sum_{x\in\tilde\lambda} c(x,0)^{r-1}+(-1)^{r}\sum_{y\in\tilde\mu} c(y,1+k-kp_0)^{r-1}.$$

If parameters $k,p_0$ are non-special, then this equivalence relation is trivial. More precisely, we have the following result \cite{SV6}.

\begin{proposition}
If $k$ is not rational and $p_0\ne n+k^{-1}m,\,m,n\in\Bbb Z_{>0},$  then  $\alpha$ is $\mathcal E$-equivalent to $\tilde\alpha$ if and only if $\alpha=\tilde\alpha.$ 
\end{proposition}
\begin{proof} 
If (\ref{eqv}) is true for all $r\ge1$ then the sequences 
$$
(c(x,0), -c(y,1+k-kp_0))_{x\in\lambda,y\in\tilde\mu},\quad (c(x,0), -c(y,1+k-kp_0))_{x\in\tilde\lambda,y\in\mu}
$$
coincide  up to a permutation. Therefore we have for every $x\in \lambda$ two possibilities: $c(x,0)=c(\tilde x,0)$ for some $\tilde x\in\tilde\lambda$, or $c(x,0)=-c(\tilde y,1+k-kp_0)$ for some $\tilde y\in\mu$. In the first case we have for
 $x=(ij),\,\tilde x=(\tilde i\tilde j)$ the relation $j-\tilde j+k(i-\tilde i)=0$,
so $j=\tilde j,\,i=\tilde i$ since $k$ is not rational. 

In the second case we have for $
 \tilde y=(\tilde i\tilde j)$ that 
 \begin{equation}\label{sec}
 kp_0=j+\tilde j-1+k(i+\tilde i-1),
 \end{equation}
which contradicts to our assumption, since both $j+\tilde j-1$ and $i+\tilde i-1$ are positive integers.
\end{proof}

Consider now the case of special values of parameters when $$p_0=n+k^{-1}m$$
for some $n,m\in\Bbb Z_{>0},$ still assuming that $k$ is not rational.
Denote by $\pi(n,m)$ the rectangular Young diagram of size $n\times m$ and the corresponding bipartition $\pi=(\pi(n,m),\pi(n,m))$. 
Define the central symmetry transformation $\theta$ acting on $(ij)\in \pi(n,m)$ by $$\theta(ij)=(n-i+1,m-j+1).$$ 
Inclusion of the Young diagrams induces the following partial order on bipartitions. We say that  $\alpha\subset \tilde\alpha$ if and only if 
$\lambda\subset\tilde\lambda$ and $\mu\subset\tilde\mu$, where the Young diagrams are understood as the subsets of the plane.
We will use the same convention for all set-theoretical operations for bipartitions.

\begin{proposition} 
 Bipartition $\alpha=(\lambda,\mu)$ is $\mathcal E$-equivalent to $\tilde \alpha=(\tilde\lambda,\tilde\mu)$ if and only if
 \begin{equation}\label{sec1}
\alpha\setminus\pi=\tilde\alpha\setminus\pi
 \end{equation}
 and
 \begin{equation}\label{sec2}
\theta(\lambda\setminus \tilde \lambda)=\mu\setminus \tilde \mu, \quad \theta(\tilde\lambda\setminus  \lambda)=\tilde\mu\setminus\mu.
 \end{equation}
 \end{proposition}

\begin{proof}
We will use the notations from the proof of the previous proposition. 
If $\alpha$ is equivalent to $\tilde \alpha,$ then for any $x=(ij) \in \lambda\setminus \pi(n,m)$ there is only the first possibility and therefore $x \in \tilde \lambda\setminus \pi(n,m).$ Thus $\lambda\setminus \pi(n,m) \subset \tilde \lambda\setminus \pi(n,m),$ and by symmetry $\lambda\setminus \pi(n,m) = \tilde \lambda\setminus \pi(n,m).$ Similarly we have $\mu\setminus \pi(n,m) = \tilde \mu\setminus \pi(n,m)$ and (\ref{sec1}).

From (\ref{sec1}) it follows that $\lambda\setminus \tilde \lambda$ is contained in $\pi(n,m).$ For $x=(ij) \in \lambda\setminus \tilde \lambda$ there exists only second possibility, which means that there exists
$\tilde y=(\tilde i\tilde j)\in \mu$ such that 
 $j+\tilde j-1+k(i+\tilde i-1)=kp_0=n+km.$ Since $k$ is not rational, this implies that
 $$j+\tilde j-1=n, \quad i+\tilde i-1=m,$$
 which means that $\theta(x) \in \mu \setminus \tilde \mu.$
 Similarly we have $\theta(\mu \setminus \tilde \mu) \subset \lambda\setminus \tilde \lambda.$ Since $\theta$ is an involution, this implies 
 $$\theta(\lambda\setminus \tilde \lambda)=\mu\setminus \tilde \mu.$$ By symmetry we have 
 $\theta(\tilde\lambda\setminus  \lambda)=\tilde\mu\setminus\mu.$

Conversely, assume that we have the relations (\ref{sec1}), (\ref{sec2}).
We have to show that the sequences 
$$
(c(x,0), -c(y,1+k-kp_0))_{x\in\lambda,y\in\tilde\mu},\quad (c(x,0), -c(y,1+k-kp_0))_{x\in\tilde\lambda,y\in\mu}
$$
coincide  up to a permutation. We have the disjoint unions $$\lambda=(\lambda \setminus \pi(n,m))\cup (\lambda \setminus \tilde \lambda)\cup (\lambda\cap\tilde\lambda\cap \pi(n,m)),$$
$$\tilde\mu=(\tilde\mu \setminus \pi(n,m))\cup (\tilde\mu \setminus \mu)\cup (\tilde\mu\cap\mu\cap \pi(n,m)),$$
$$\tilde\lambda=(\tilde\lambda \setminus \pi(n,m))\cup (\tilde\lambda \setminus \lambda)\cup (\lambda\cap\tilde\lambda\cap \pi(n,m)),$$
$$\mu=(\mu \setminus \pi(n,m))\cup (\mu \setminus \tilde \mu)\cup (\mu\cap\tilde\mu\cap \pi(n,m)).$$
Using this, the relations (\ref{sec1}), (\ref{sec2}) and the identity
\begin{equation}\label{ident}
c(\theta(x), 1+k-kp_0)=(m+kn-kp_0)-c(x,0), \,\, x \in \pi(n,m)
\end{equation}
we can identify the corresponding contributions in these sequences and have the result.
\end{proof}

Consider the set $\mathcal P_{n,m}$ of bipartitions $\alpha \subset \pi=(\pi(n,m),\pi(n,m)).$ 
For such partitions the equivalence relation can be described in the following simple way. 
Introduce the involution $\omega: \mathcal P_{n,m} \rightarrow \mathcal P_{n,m}$ such that for $\alpha=(\lambda,\mu)$
\begin{equation}\label{omega}
\omega (\alpha)=(\lambda, \pi(n,m) \setminus \theta(\mu)).
 \end{equation}

Introduce now another equivalence relation $\mathcal R$ on bipartitions. We say that $\alpha= (\lambda,\mu)$ is {\it $\mathcal R$-equivalent} $\tilde \alpha=(\tilde\lambda,\tilde\mu)$ if 
\begin{equation}\label{secR}
\lambda\cap\mu=\tilde\lambda\cap\tilde\mu, \,\,\, \lambda\cup\mu=\tilde\lambda\cup\tilde\mu.
 \end{equation}
 
 \begin{thm} \label{omegath}
On the set $\mathcal P_{n,m}$ the involution (\ref{omega}) transforms the equivalence relation $\mathcal E$ into $\mathcal R.$
 \end{thm}

\begin{proof} 
Let $\alpha=(\lambda,\mu)$ be $\mathcal E$-equivalent to $\tilde\alpha=(\tilde\lambda,\tilde\mu).$ 
It is enough to prove that
\begin{equation}\label{sec3}
\lambda \cup (\pi(n,m)\setminus\theta(\mu))=\tilde\lambda \cup (\pi(n,m)\setminus\theta(\tilde\mu))
 \end{equation}
and 
\begin{equation}\label{sec4}
\lambda \cap (\pi(n,m)\setminus\theta(\mu))=\tilde\lambda \cap (\pi(n,m)\setminus\theta(\tilde\mu)).
 \end{equation}
 Let's prove (\ref{sec4}). Let $x \in \lambda \cap (\pi(n,m)\setminus\theta(\mu)),$ then 
 $x \in \lambda$ and $x \notin \theta(\mu).$ 
 
 Assume that $x \notin \tilde \lambda$, then from (\ref{sec2}) it follows that $\theta(x)\in \mu$ and thus $x \in \theta(\mu).$ Contradiction means that $x \in \tilde\lambda.$ 
 
 Assume now that $x \notin \pi(n,m)\setminus \theta(\tilde\mu),$ which means that $x \in \theta(\tilde\mu).$ Since $x \notin \theta(\mu)$ we have $x \in \theta(\tilde\mu\setminus\mu).$ Using the second part of (\ref{sec2}) we see that $x \in \tilde \lambda\setminus \lambda$ and hence $x \notin \lambda$, which is a contradiction. Now (\ref{sec4}) follows from the symmetry between $\alpha$ and $\tilde \alpha.$
 The proof of (\ref{sec3}) is similar. 
 
 This proves that $\mathcal E$-equivalence implies $\mathcal R$-equivalence for $\omega$-transformed bipartitions.
 The converse claim can be proved in a similar way.
  \end{proof} 

This can be used to describe the structure of $\mathcal E$-equivalence classes of bipartitions from $\mathcal P_{n,m}$.

\begin{thm}\label{weight} Let $\alpha \in \mathcal P_{n,m}$ and $E$ be its $\mathcal E$-equivalence class. Then the following holds true:

$1)$ $E \subset \mathcal P_{n,m}.$

$2)$  $E$ contains the minimal and maximal bipartitions $\alpha_{m}, \, \alpha_M$ such that
$$
\alpha_m\subset\alpha\subset \alpha_M
$$
for any bipartition $\alpha\in E$. They can be characterised by the properties $\lambda\cap\theta(\mu)=\emptyset$ and by $\lambda \cup \theta(\mu)=\pi(n,m)$ respectively.

$3)$ Let $\alpha_m=(\lambda_m, \mu_m),\, \alpha_M=(\lambda_M, \mu_M)$ and 
\begin{equation}\label{eqv1}
\lambda_M\setminus\lambda_m=\nu_1\cup\nu_2\cup\dots\cup\nu_r,\,\,\mu_M\setminus\mu_m=\tau_1\cup\tau_2\cup\dots\cup\tau_s
\end{equation}
be the decomposition of the corresponding skew diagrams into connected components. Then $\nu_i,\tau_j\subset\pi(n,m)$, $r=s$ and, after a reordering,  $$\theta(\nu_i)=\tau_i,\,\,i=1,2,\dots,r.$$

$4)$  Every element  $\alpha$ from $E$ can be represented  uniquely in the form 
\begin{equation}\label{eqx}
\alpha=\alpha_m\cup(\nu_{a_1},\tau_{a_1})\cup(\nu_{a_2},\tau_{a_2})\cup\dots\cup(\nu_{a_l},\tau_{a_l}),
\end{equation}
where $\{a_1,\,a_2,\dots,a_l\}$ is a subset  of $\{1,\,2,\dots,r\}.$
Any set of this form is a bipartition from $E,$ so the equivalence  class $E$ contains $2^r$ elements.
\end{thm}

\begin{proof} The first part follows immediately from (\ref{sec1}). Applying the involution $\omega$ and the previous theorem we have the remaining claims using simple geometric analysis of the corresponding Young diagrams (see Fig. 1).
\end{proof}

\begin{figure}
\centerline{ \includegraphics[width=12cm]{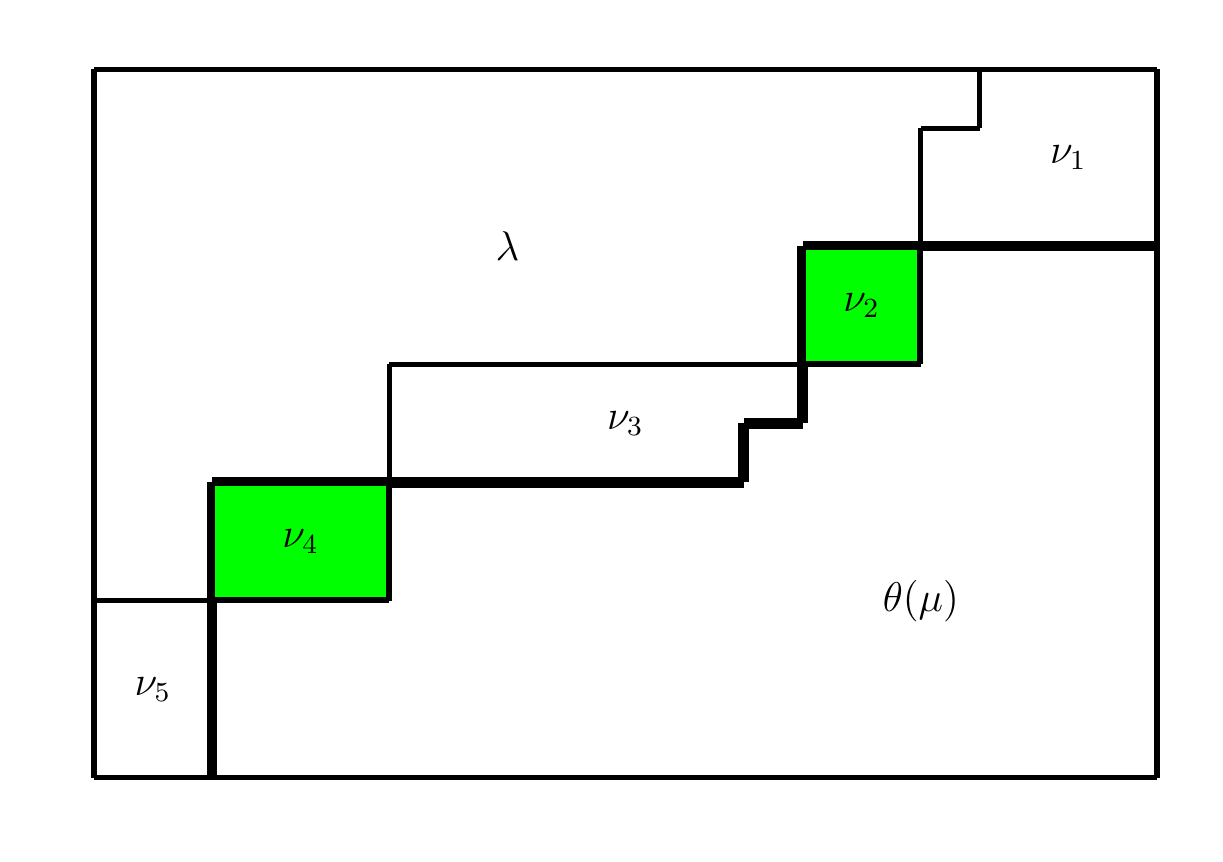}} \label{Geom}
\caption{Intersection of $\lambda$ and $\theta(\mu)$ (shaded) in the rectangle $\pi((n,m)$ and the corresponding connected components $\nu_i.$ The boundary of $\theta(\mu)$ is shown in bold.} \end{figure}

To describe $\mathcal E$-equivalence class for general  bipartition $\alpha=(\lambda,\mu)$ denote by $\alpha_{\pi}$ the bipartition $\alpha_{\pi}=\alpha\cap\pi=(\lambda_{\pi},\mu_{\pi}):$
$$
(\lambda_{\pi},\mu_{\pi})=(\lambda\cap\pi(n,m),\mu\cap\pi(n,m)).
$$ 

\begin{corollary} \label{cor1} Let $E(\alpha_{\pi})$ be the $\mathcal E$-equivalence class of $\alpha_{\pi}$. Then 
$\mathcal E$-equivalence class of $\alpha$ can be described as
$$
E(\alpha)=\{\gamma=\beta \cup (\alpha\setminus\pi)\in\mathcal P\times\mathcal P \mid \beta\in E(\alpha_{\pi})\}.
$$ 
$E(\alpha)$ contains the minimal and maximal bipartitions $\alpha_{m}, \, \alpha_M$ such that
$$
\alpha_m\subset\alpha\subset \alpha_M
$$
with parts $3)$ and $4)$ of theorem \ref{weight} remaining valid for any bipartition $\alpha.$ 
\end{corollary}

Note that $E(\alpha) \cap \mathcal P_{n,m} \subset E(\alpha_{\pi})$ in general does not coincide with $E(\alpha_{\pi}).$

\section{Translation functors and regular basis}

In \cite{SV6} we have introduced the Jack-Laurent symmetric functions $P_{\alpha}=P^{(k,p_0)}_{\alpha}\in \Lambda^{\pm}$ indexed by bipartition $\alpha =(\lambda,\mu).$
As we have shown they are well defined provided $k$ is not rational and $p_0 \ne n+k^{-1}m$ with $n,m \in \mathbb Z_{>0}.$
Equivalently, we can consider $P_{\alpha}$ as elements of $\Lambda^{\pm}_{p_0}=\Lambda^{\pm}\otimes \mathbb C(p_0)$, where $\mathbb C(p_0)$ is the field of rational functions of $p_0.$

Now we are going to study what happens when $p_0=n+k^{-1}m$ assuming that $k,n,m$ are fixed with $k$ not rational and  $n,m \in \mathbb Z_{>0}.$ 
Then $P_{\alpha}^{(k,p_0)}$ as functions of $p_0$  may have pole at $p_0=n+k^{-1}m$ depending on the choice of bipartition $\alpha.$

The aim of this section  is to construct a basis in $\Lambda^{\pm}$ which is regular at $p_0=n+k^{-1}m$. 
More precisely, we will define the Laurent symmetric functions $Q_{\alpha}=Q^{(k,p_0)}_{\alpha}\in \Lambda^{\pm}$, which are regular at  $p_0=n+k^{-1}m$,
such that for any $\alpha$
$$
Q_{\alpha}=\sum_{\beta \in E(\alpha), \, \beta \subset \alpha} a_{\beta\alpha}P_{\beta}
$$
with some coefficients $a_{\beta\alpha}=a_{\beta\alpha}(k,p_0)$ which are rational functions of $p_0.$

In order to do this we are going to produce some family  of linear transformations   $\mathcal F_{E,F}$
acting on  $\Lambda^{\pm}_{p_0}$ which are similar to the {\it translation functors} in the representation theory \cite{BG,Z}.

Let  $E$ be an $\mathcal E$-equivalence class of bipartitions and $V_{E} \subset \Lambda^{\pm}_{p_0}$ be the linear span over $\mathbb C(p_0)$ of $P_{\alpha}$  with $\alpha\in E.$
We have the decomposition of vector spaces over $\mathbb C(p_0)$
$$
\Lambda^{\pm}_{p_0}=\bigoplus_{E}V_{E},
$$
where the sum is taken over all $\mathcal E$-equivalence classes of bipartitions.

Denote by $Pr_E$ the projector onto the subspace $V_{E}$ with respect to this decomposition and define for any $\mathcal E$-equivalence classes $E$ and $F$ the linear map
\begin{equation}\label{tf}
\mathcal F_{E,F}(f):=Pr_{F}(p_1f), \, f \in V_{E}.
\end{equation}
The next result is quite simple but very important.

\begin{proposition} \label{nonsing}Let $f\in V_{E}$ and  suppose that $f$ has no pole at $p_0=n+k^{-1}m$. Then 
for any $\mathcal E$-equivalence class $F$ the function
$
\mathcal F_{E,F}(f)
$
also has no pole at $p_0=n+k^{-1}m$.
\end{proposition}

\begin{proof} We have
\begin{equation}\label{deco}
p_1V_E\subset V_F\oplus V_{E_1}\oplus\dots\oplus V_{E_L}
\end{equation}
where $F,\,E_1,\dots, E_L$ are different classes of equivalence. First we will construct linear operator  $\mathcal C_1$ which polynomially depends on CMS integrals $\mathcal B^{(r)}$ with coefficients having no poles at $p_0=n+k^{-1}m$ and such that 
$$
\mathcal C_1(V_{E_1})=0,\,\, \mathcal C_1(v)=v,\,\, v\in V_F
$$
Let $\alpha_1,\dots,\alpha_N$ be all bipartitions in $F$ and $\beta_1,\dots,\beta_M$ all bipartitions in $E_1$. Then by definition of the equivalence classes there is  $r_1 \in \mathbb Z_{>0}$ such that 
$$
b_{r_1}(\alpha_1, k, p_0) \ne b_{r_1}(\beta_j,  k, p_0), \,\,\, j=1,\dots,M
$$
when $p_0= n+k^{-1}m.$
Let
$$
f_1(t)=\prod_{j=1}^M(t-b_{r_1}(\beta_j, k , p_0)),
$$
then operator $\mathcal D_1=f_1(\mathcal B^{(r_1)})$, where $\mathcal B^{(r)}$ are the CMS integrals from Theorem 
\ref{Bern},
acts as zero in $V_{E_1}$ and in $V_F$ as a diagonal operator 
$$
\mathcal D_1P_{\alpha_i}=g_1(\alpha_i, k, p_0)P_{\alpha_i}, \,\, i=1,\dots, N,
$$
where $g_1(\alpha_i, k, p_0)=f_1(b_{r_1}(\alpha_i, k, p_0)).$
Now having in mind Cayley-Hamilton theorem we can define 
$$
\mathcal C_1=(-1)^{N+1}\frac{1}{\sigma_N}\left(\mathcal D_1^{N}-\sigma_1\mathcal D_1^{N-1}+\dots+(-1)^{N-1}\sigma_{N-1}\mathcal D_1\right)
$$
where $\sigma_1,\dots,\sigma_N$ stand for the elementary symmetric polynomials in 
$$g_1(\alpha_1,k, p_0), \,\dots \, ,g_1(\alpha_N,k, p_0).$$
 From our assumptions we see that 
$\sigma_N=g_1(\alpha_1, k, p_0)\dots g_1(\alpha_N, k, p_0)\ne 0$
when $p_0=k^{-1}n+m.$
We see that   $\mathcal C_1(V_{E_1})=0$ and  by the Cayley-Hamilton theorem 
$\mathcal C_1$ acts as the identity in $V_{F}$. 

In the same way we can construct operators $\mathcal C_2,\dots\mathcal C_L$ and define
$$
\mathcal C=\mathcal C_1\mathcal C_2\dots\mathcal C_L.
$$
Let 
$
p_1f=g+g_1+\dots+g_L
$
be the decomposition according to (\ref{deco}).
Applying to both sides of this equality the operator $\mathcal C$ we get 
$$
\mathcal C(p_1f)=g = Pr_F(p_1f).
$$
But, since $\mathcal B^{(r)}$ are polynomial in $p_0,$ $\mathcal C$ is a differential operator with coefficients that have no poles at $p_0=n+k^{-1}m$, so  both sides must be regular at this point.
\end{proof}

The following definition is motivated by the Pieri formula for Jack--Laurent symmetric functions \cite{SV6}. Let $\alpha =(\lambda,\mu) \in \mathcal P_{n,m}$ be a bipartition inside $\pi.$

For any box $x\in\pi(n,m)$  
 define the set of bipartitions $S_x(\alpha)$ as  
 $$
 S_x(\alpha)=\{(\lambda\cup x,\mu),\,(\lambda,\mu\setminus \theta(x))\}
 $$
assuming that $x\notin\lambda$  and $ \lambda\cup x$ is a Young  diagram, and that $\theta(x) \in \mu$ and $\mu\setminus \theta(x)$ is a Young diagram (otherwise the corresponding element is dropped from the set).

Let us denote by $X(\alpha)$ the set of all bipartitions in the right hand side of the Pieri formula (see formula (56) from \cite{SV6}):
$X(\alpha)$ is the set of all bipartions $\beta=(\tilde\lambda, \tilde \mu)$ 
such that $\alpha$ can be obtained from $\beta$ by deleting a box from $\tilde\lambda$ or adding a box to $\tilde\mu.$

\begin{proposition} \label{Eqv} Let $E$ be an $\mathcal E$-equivalence class and suppose that there is $\alpha\in E$ such that   $S_x(\alpha)$  is not empty. Then there exists a unique $\mathcal E$-equivalence class $E_x$  different from $E$ such that  for any $\alpha\in E$ 
$$
X(\alpha)\cap E_x=S_x(\alpha).
$$
\end{proposition}
\begin{proof}  
Let us prove first that if $\alpha$ is $\mathcal E$-equivalent to $\tilde\alpha$ then $S_x(\alpha)$ and 
$S_x(\tilde\alpha)$ belong to the same  $\mathcal E$-equivalence class. Applying the involution $\omega$ we reduce this to the following statement. Let $\omega(\alpha)=(\lambda,\mu), \omega(\tilde\alpha)=(\tilde\lambda,\tilde\mu)$ and 
$$\lambda\cup\mu=\tilde\lambda\cup\tilde\mu, \, \lambda\cap\mu=\tilde\lambda\cap\tilde\mu.$$
Without loss of generality we can assume that the box $x$ can be added to $\lambda$ and $\tilde\lambda$. We need to prove that 
$$(\lambda\cup x)\cup\mu=(\tilde\lambda\cup x)\cup\tilde\mu, \, \, \, (\lambda\cup x)\cap\mu=(\tilde\lambda\cup x)\cap\tilde\mu.$$
The first equality is obvious. To prove the second consider two cases: $x \notin\mu$ and $x \in\mu.$

If $x \notin\mu$ then $x \notin\lambda\cup\mu=\tilde\lambda\cup\tilde\mu,$ hence $x \notin \tilde\mu,$
which implies that $(\lambda\cup x)\cap\mu=(\tilde\lambda\cup x)\cap\tilde\mu.$

If $x\in \mu$ then $x \in\lambda\cup\mu=\tilde\lambda\cup\tilde\mu,$ and hence $x \in \tilde\mu.$
Therefore $$(\lambda\cup x)\cap\mu=\lambda\cap\mu=\tilde\lambda\cap\tilde\mu=(\tilde\lambda\cup x)\cap\tilde\mu.$$
Hence there exists a unique equivalence class $E_x$ containing the union of $S_x(\alpha), \, \alpha \in E.$ The relation $X(\alpha)\cap E_x=S_x(\alpha)$ is easy to check.

We only left to prove that these equivalence classes $E$ and $E(x)$ are different. Suppose that $(\lambda,\mu)$ and $(\lambda\cup x,\mu)$ are $\mathcal R$-equivalent. Then we have 
 $$
 \lambda\cup\mu=\lambda\cup x\cup\mu,\,\,  \lambda\cap\mu=(\lambda\cup x)\cap\mu,
 $$
implying that $x \in \lambda,$ which is a contradiction.  \end{proof}

For any box $x\in\pi(n,m)$  
 define now the set of bipartitions $S^x(\alpha)$ as  
 $$
 S^x(\alpha)=\{(\lambda\setminus x,\mu),\,(\lambda,\mu\cup \theta(x))\}.
 $$
 In the same way as in proposition \ref{Eqv} it can be proven that 
 there exists a unique
 $\mathcal E$-equivalence class $E^x$, which contains $S^x(\alpha)$ for any $\alpha \in E.$
 
Let $x\in\pi(n,m)$. Denote by $\mathcal F_x$ the linear transformation defined by
$$
\mathcal F_x=\mathcal F_{E, E_x}.
$$

The following proposition is based on the Pieri formula for Jack-Laurent symmetric functions \cite{SV6}.
Introduce the following functions for bipartition $\alpha=(\lambda,\mu)$ and box $x=(ij)$:
\begin{equation}\label{U}
U(x,\alpha; p_0)=U_1(x,\alpha)U_2(x,\alpha; p_0)U_3(x,\alpha; p_0),
\end{equation}
\begin{equation}\label{Ux}
U_1(x,\alpha)=\prod_{r=i+1}^{l(\mu)}\frac{c_{\mu}(jr,1+k)c_{\mu}(jr,-k)}{c_{\mu}(jr,1)c_{\mu}(jr,0)},
\end{equation}
\begin{equation}\label{Uy}
U_2(x,\alpha; p_0)=\prod_{r=1}^{l(\lambda)}\frac{c_{\alpha}(jr,-1-k(p_{0}+2))c_{\alpha}(jr,-kp_{0})}{c_{\alpha}(jr,-1-k(p_{0}+1))c_{\alpha}(jr,-k(p_{0}+1))},
\end{equation}
\begin{equation}\label{Uz}
U_3(x,\alpha; p_0)=\frac{(j-1+k(l(\lambda)+\mu^{\prime}_{j}-p_{0}-1))(j+k(\mu^{\prime}_{j}-l(\mu)))}
{(j+k(l(\lambda)+\mu^{\prime}_{j}-p_{0}))(j-1+k(\mu^{\prime}_{j}-l(\mu)-1))},
\end{equation}
where
$$
c_{\lambda}(jr,a)=\lambda_{r}-j-k(\lambda^{\prime}_{j}-r)+a,\;\;
$$
$$
c_{\alpha}(jr,a)=\lambda_{r}+j+k(\mu^{\prime}_{j}+r)+a,
$$
and $\lambda^{\prime}$ as before is the Young diagram conjugated (transposed) to $\lambda$.

\begin{proposition} \label{FX}
The action of $\mathcal F_x$ on Jack-Laurent symmetric functions can be described by
\begin{equation}\label{dec}
\mathcal F_x(P_{\lambda,\mu})=V(x,\lambda,\mu)P_{\lambda\cup x,\mu}+U(\theta(x),\lambda,\mu; p_0)P_{\lambda,\mu \setminus \theta(x)},
\end{equation}
where 
\begin{equation}\label{Vx}
V(x,\lambda,\mu)=\prod_{r=i+1}^{l(\mu)}\frac{c_{\mu}(jr,1+k)c_{\mu}(jr,-k)}{c_{\mu}(jr,1)c_{\mu}(jr,0)},\,\,x=(ij)
\end{equation}
and $U(x,\lambda,\mu; p_0)$ is defined by (\ref{U}).
\end{proposition}

\begin{proof}This follows immediately from proposition \ref{Eqv} and Pieri formula for Jack--Laurent symmetric functions \cite{SV6}.
\end{proof}

\begin{lemma}\label{main}
Let us assume that the box $\theta(x)=(n-i+1, \, m-j+1), \, x=(ij)$ can be removed from $\mu,$ then the following hold true:

$1)$ If $\lambda_{i-1}=j-1$, or $\lambda_{i+1}=j$ then the numerator of the function $U(\theta(x),\lambda,\mu; p_0)$ has zero of the first order at $p_0=n+k^{-1}m$;

$2)$ If $\lambda_{i}=j$, or $\lambda_{i}=j-1$, or $j=1, \, i=l(\lambda)+1$, then the denominator of the function $U(\theta(x),\lambda,\mu; p_0)$ has zero of the first order at $p_0=k^{-1}n+m.$

In all other cases neither numerator nor denominator of  $U(\theta(x),\lambda,\mu; p_0)$ has zero at $p_0=k^{-1}n+m.$
\end{lemma}

\begin{proof} 
Note that $U_1$ does not depend on $p_0.$
Introduce the new variable $\delta=n+km-kp_0.$ Since the box $\theta(x)=(i'j')$ can be removed from $\mu$ we  have $\mu'_{j'}=i'=n-i+1$ and 
$$c_{\alpha}(j'r,a-kp_0)=\lambda_{r}+j'+k(\mu^{\prime}_{j'}+r)+a-kp_0=\lambda_r-j-k(i-1-r)+\delta +1 +a.$$
The second factor $c_{\alpha}(jr,-kp_{0})$ in the numerator of $U_2$ corresponds to $a=0$ and thus equals to $\lambda_r-j-k(i-1-r)+\delta +1.$ Since $k$ is assumed not rational the condition $\delta =0$ gives $r=i-1$ and $\lambda_r=j-1$ and thus $\lambda_{i-1}=j-1$, which is the first condition in case $1)$. Similarly one can check the rest.
\end{proof}

Let $E$ be an $\mathcal E$-equivalence class consisting of more than one element and $(\lambda_M,\mu_M)$, $(\lambda_m,\mu_m)$ be the  maximal and the minimal bipartitions in it. Let us choose $x\in \lambda _M\setminus\lambda_m$  such that $\lambda_M\setminus x$ is a partition and let  $\nu$ be the connected component containing $x$.  Let  $\alpha=(\lambda,\mu)\in E$ then it is easy to check that  $\mu\cup\theta(x)$ is a partition if and only if $\lambda\cap\nu=\emptyset$. Therefore for any $\alpha\in E$ we can define a map $\psi : E \rightarrow  E^x$ by
\begin{equation}\label{psi}
\psi(\alpha)=\begin{cases}(\lambda\setminus x,\mu),\,\,\nu\subset \lambda \\(\lambda,\mu\cup\theta(x)),\,\,\nu\cap \lambda=\emptyset. \end{cases}
\end{equation}
It is easy to see that $\psi$ preserves the inclusions of bipartitions.

\begin{lemma}\label{cases} 
The following statements hold true:

$(1)$ If $\nu\setminus x$ is nonempty and connected then $\psi$ is a bijection  and  for any $\alpha\in E$
$$
\mathcal F_x(P_{\psi(\alpha)})=d(x,p_0,\alpha)P_{\alpha},
$$
where $d(x,p_0,\alpha)$ is nonzero rational function in $p_0$ which has neither zero nor pole at $p_0=n+k^{-1}m$.

$(2)$ If $\nu\setminus x=\nu_1\cup\nu_2$ is nonempty and  not connected then  $\psi$ is injective and for any $\alpha\in E$
$$
\mathcal F_x(P_{\psi(\alpha)})=d(x,p_0,\alpha)P_{\alpha},
$$
where $d(x,p_0\alpha)$ has  zero of the first order at $p_0=n+k^{-1}m$ if $\lambda\cap\nu=\emptyset$ and $d(x,p_0,\alpha)$ has  neither zero nor pole at $p_0=n+k^{-1}m$ if $\lambda\supset\nu$. If $\gamma\in E^x$ and $\gamma\notin Im\,\psi$ then
$
\mathcal F_x(P_{\gamma})=0.
$

$(3)$  If $\nu\setminus x=\emptyset$ is empty then  $\psi$ is surjective such that for any $\gamma\in E^x$ 
$$
\psi^{-1}(\gamma)=\{ \alpha,\alpha\cup (x,\theta(x))\}
$$
 and 
$$
\mathcal F_x(P_{\gamma})=d(x,p_0,\alpha\cup x)P_{\alpha\cup (x,\theta(x))}+d(x,p_0,\alpha)P_{\alpha},
$$
where $d(x,p_0,\alpha\cup x)$ has neither zero nor pole and $d(x,p_0,\alpha)$  has a pole of the first order  at $p_0=n+k^{-1}m.$ 
\end{lemma}
\begin{proof}
Let $x=(ij).$ 
Consider the case $(1).$ 
If $\psi(\alpha)=(\lambda\setminus x,\mu)$ then $d(x,p_0)=V(x,\psi(\alpha))$
and the claim follows. If $\psi(\alpha)=(\lambda,\mu\cup\theta(x))$ then $d(x,p_0)=U(\theta(x),\psi(\alpha); p_0).$ We claim that $U(\theta(x),\psi(\alpha); p_0)$ has no zero or pole at $p_0=k^{-1}n+m.$ Indeed, according to lemma (\ref{main}) we should show that none of the relations in the lemma are satisfied. The last relation $j=1, \, i=l(\lambda)+1$ is impossible since $\nu\setminus x$ is nonempty.
To check the rest note that since $\nu\cap \lambda=\emptyset$ we have $j=\lambda_i+\nu_i.$
If $\lambda_i=j$ then $\nu_i=0$ which is impossible. If $\lambda_{i+1}=j$ we have $\lambda_i\geq \lambda_{i+1}=j,$ which implies $\nu_i \leq 0$, which is also impossible.
If $\lambda_i=j-1$ and $\lambda_{i-1}=j-1$ simultaneously then the zero in the denominator cancels the zero in the numerator and we have the claim. If at the least one of these relations are not valid then we have the strict inequality $\lambda_{i-1}>\lambda_i.$ Let $\lambda_{i-1}=j-1$, then $\lambda_{i-1}=\lambda_i+\nu_i-1,$ which implies that $\nu_i>1$. It is easy to see that this contradicts to the connectivity assumption of $\nu\setminus x.$ The last case to check is when $\lambda_i=j-1, \lambda_{i-1}>\lambda_i.$ This case contradicts to the connectivity of $\nu.$ 
This proves the lemma in case (1).
The remaining cases can be proved in the same way.
\end{proof}

\begin{thm}  
\label{regul}
Let $\alpha\in \mathcal P_{n,m}$ and $E$ be the $\mathcal E$-equivalence class containing $\alpha$, $k \notin \mathbb Q$ be fixed. Then  there are rational functions  $a_{\beta\alpha}(p_0)$  with $\beta\in E,\,\, \beta\subset \alpha$ such that $a_{\alpha\alpha}=1$  and $a_{\beta\alpha}(p_0)$ has a pole at $p_0=n+k^{-1}m$ of order, which is  equal to the number of connected components in $\alpha\setminus \beta,$ and such that the linear combination of Jack-Laurent symmetric functions
$$
Q_{\alpha}=\sum_{\beta\in E,\,\, \beta\subset \alpha}a_{\beta\alpha}(p_0)P_{\beta}
$$
is regular at $p_0=n+k^{-1}m.$
\end{thm}

\begin{proof} Let us prove theorem by induction on $|\lambda\setminus\lambda_m|$.

If $|\lambda\setminus\lambda_m|=0$ then   $\alpha=\alpha_m =(\lambda,\mu)$ and in this case the theorem states that $P_{\alpha_m}$ is regular at $p_0=k^{-1}n+m.$ By part $2)$  of theorem \ref{weight}  we have $\lambda\cap\theta(\mu)=\emptyset$. Let $x_1,\dots,x_N$ be all boxes of the diagram $\lambda$ beginning from the first box of the first row and ending by the last box of the last row. Consider the following function
$$
Q=\mathcal F_{x_N}\dots\mathcal F_{x_1}(P_{\emptyset,\mu}).
$$
We have  $P_{\emptyset,\mu}=P^*_{\mu}$ is dual to Jack symmetric function and thus does not depend on $p_0$. Therefore by proposition \ref{nonsing} $Q$ has no poles at $p_0=n+k^{-1}m.$ Moreover, since $\lambda\cap\theta(\mu)=\emptyset$ by proposition \ref{FX} and lemma \ref{main}  
we have
$$
Q=V_N\dots V_1 P_{\alpha_m},
$$
where $V_i,\, i=1,\dots, N$ do not depend on $p_0$, and thus $P_{\alpha_m}$ is regular at $p_0=n+k^{-1}m.$

Now suppose  that $\alpha\in E$ and $\alpha\ne\alpha_m$. Therefore there exists a connected component $\nu\subset\lambda_M\setminus \lambda_m,\, \nu \subset \lambda$. Let us pick $x\in\nu$ such that $\lambda\setminus x$ is a Young diagram. It is easy to see that $\lambda_M\setminus x$ is also a Young diagram. 

Consider three different possibilities as in lemma \ref{cases}.
In all three cases
$$
\psi(\alpha)=(\lambda\setminus x,\mu),\quad \psi(\alpha_m)=(\lambda_m,\mu\cup \theta(x))
$$ 
and $\lambda_m\cap( \theta(\mu)\cup x)=\emptyset$. Therefore  $\psi(\alpha_m)$ is the minimal element in $E^x$ and $\mid(\lambda\setminus x)\setminus\lambda_m\mid=\mid\lambda\setminus\lambda_m\mid-1$ and  we can apply inductive assumption. After applying  $\mathcal F_x$ to $Q_{\psi(\alpha)}$  and using lemma \ref{cases} we get

$$
\mathcal F_x(Q_{\psi(\alpha)})=\sum_{\beta\in E,\,\, \beta\subset \alpha}\tilde a_{\beta\alpha}(p_0)P_{\beta}
$$
with some coefficients $a_{\beta\alpha}(p_0)$ which  are rational functions in $p_0$. By 
proposition \ref{nonsing} $\mathcal F_x(Q_{\psi(\alpha)})$ is non-singular at $p_0=n+k^{-1}m$  and $d(x,p_0,\alpha)$ is also non-singular  and non-vanishing by lemma \ref{cases} in all three cases. Define 
$$
Q_{\alpha}=\frac{1}{d(x,p_0,\alpha)}\mathcal F_x(Q_{\psi(\alpha)})=\sum_{\beta\in E,\,\, \beta\subset \alpha}a_{\beta\alpha}(p_0)P_{\beta}
$$
with $a_{\beta\alpha}(p_0)=\tilde a_{\beta\alpha}(p_0)/d(x,p_0,\alpha).$

Now let us prove that the coefficients $a_{\beta\alpha}(p_0)$ have the analytic properties stated in the theorem. We have in all cases 
$$
a_{\alpha\alpha}=\frac{\tilde a_{\alpha\alpha}}{d(x,p_0,\alpha)}=\frac{d(x,p_0,\alpha)}{d(x,p_0,\alpha)}=1.
$$
Let $\beta\ne\alpha$. Then again in all three cases from lemma \ref{cases} one can see that
$$
a_{\beta\alpha}=\frac{d(x,p_0,\beta)}{d(x,p_0,\alpha)}a_{\psi(\beta)\psi(\alpha)}.
$$

 Now consider three different cases separately. 

$1)$ If $\nu\setminus x$ is non empty and connected  then by the first statement of lemma \ref{cases} $d(x,p_0,\beta)$ is regular at $p_0=n+k^{-1}m$ and the number of connected components $\alpha\setminus\beta$ is the same as the number of connected components of $\psi(\alpha)\setminus\psi(\beta)$. This implies  the theorem in this case.

$2)$ Let $\nu\setminus x=\nu_1\cup\nu_2$ be a disjoint  union of two non empty components. Consider two cases: $\beta\supset\rho$ and $\beta\cap\rho=\emptyset$, where $\rho=(\nu, \theta(\nu))$. In the first case
the number of connected components $\alpha\setminus\beta$ is the same as the number of connected components of $\psi(\alpha)\setminus\psi(\beta)$ , $d(x,p_0,\beta)$ is regular and theorem follows. In the second case the number of connected components $\alpha\setminus\beta$ is less by $1$  than the number of connected components of $\psi(\alpha)\setminus\psi(\beta)$ , $d(x,p_0,\beta)$ has zero of the first order and the theorem again  follows.

$3)$  Let $\nu= x$  and $(x, \theta(x))\in\beta$ then
  the number of connected components $\alpha\setminus\beta$ is the same as the number of connected components of $\psi(\alpha)\setminus\psi(\beta)$ , $d(x,p_0,\beta)$ is regular and the theorem follows. 
 
 If $(x, \theta(x))\notin\beta$ then  the number of connected components $\alpha\setminus\beta$ is greater by $1$  than the number of connected components of $\psi(\alpha)\setminus\psi(\beta)$, $d(x,p_0,\beta)$ has a pole of the first order and theorem again  follows. This completes the proof.
\end{proof}

\begin{corollary}
\label{corol}
The Jack--Laurent symmetric function $P^{k, p_0}_{\alpha}$ as a function of $p_0$ has a pole at $p_0=n+k^{-1}m$ of order $l$, where $l$ is defined by (\ref{eqx}) and Corollary \ref{cor1}.
\end{corollary}

For bipartitions $\alpha \in \mathcal P_{n,m}$ the order $l$ of the pole at $p_0=n+k^{-1}m$ can be described geometrically as the number of connected components in the intersection $\lambda$ and $\theta(\mu)$ (which are shaded parts in Fig. 1).

From Corollary \ref{cor1} using the same technique one can show that the assumption $\alpha\in \mathcal P_{n,m}$ in the theorem can be omitted.

\begin{proposition} Theorem \ref{regul} is true without assumption $\alpha\in \mathcal P_{n,m}$.
\end{proposition}

\section{Algebra of integrals in generalised eigenspaces}

Assume now that $k$ is non-algebraic and that $p_0=n+k^{-1}m$ for some $n,m \in \mathbb Z_{>0}$ as before.

Let $E$ be an $\mathcal E$-equivalence class of bipartitions, consisting of $2^r$ elements 
and consider $2^r$-dimensional subspace $V_E(p_0) \subset \Lambda^{\pm}$  
defined as the linear span of Jack--Laurent symmetric functions $P^{(k,p_0)}_{\alpha}, \,\, \alpha \in E$ for non-special $p_0$, and as the linear span of 
$Q^{(k,p_0)}_{\alpha}, \,\, \alpha \in E$ for all $p_0$ in a neighbourhood of $p_0=n+k^{-1}m$.

The action of the algebra of CMS integrals $\mathcal D(k, p_0)$ is diagonalisable for non-special $p_0,$
but at $p_0=n+k^{-1}m$ it has a generalised eigenspace $V_E=V_E(n+k^{-1}m)$ spanned by $Q^{(k,n+k^{-1}m)}_{\alpha}, \,\, \alpha \in E.$
We are going to study now the action of the algebra in this invariant subspace.

Consider the natural homomorphism
 $$
\varphi :\mathcal D(k,n+k^{-1}m)\longrightarrow End \,(V_{E}).$$

 \begin{thm} If $k$ is non-algebraic then the  image of the homomorphism $\varphi$  is isomorphic to the tensor product of $r$ copies of dual numbers
 $$
 \mathfrak A_r=\Bbb C[\varepsilon_1,\varepsilon_2,\dots,\varepsilon_r]/(\varepsilon_1^2,\,\varepsilon_2^2,\dots,\varepsilon_r^2).
 $$
 $V_E$ is the regular representation of $\mathfrak A_r$ with respect to this action.
 \end{thm}

\begin{proof} 
  Let $\nu_1, \dots, \nu_r$ be the corresponding sets from Theorem \ref{weight} and Corollary \ref{cor1}, describing the equivalence class $E$
and define $$
g_s(\nu)=\sum_{x\in \nu} c(x,0)^{s-1},
$$ 
 where $c(x,a)=(j-1)+k(i-1)+a$ as before.
 
 \begin{lemma}
 \label{lemm}
 If $k$ is non-algebraic then the determinant
 $$
\Delta=\left|\begin{array}{cccc}
  g_1(\nu_1)&g_{1}(\nu_2)&\ldots &g_1(\nu_r)\\
  g_2(\nu_1)&g_{2}(\nu_2) & \ldots&g_2(\nu_r)\\
  \vdots &\vdots &\vdots & \vdots\\
  g_r(\nu_1)&g_{r}(\nu_2)&\ldots &g_r(\nu_r)
 \end{array}\right|\
$$
is non zero.
 \end{lemma}

\begin{proof} Indeed, we can represent this determinant as a sum over all sequences of boxes $x_1\in\nu_1,\dots, x_r\in\nu_r$
 $$
 \Delta=\sum_{(x_1,\dots,x_r)}\Delta(x_1,\dots,x_r),$$ 
where 
$$
\Delta(x_1,\dots,x_r)=\left|\begin{array}{cccc}
  1&1&\ldots &1\\
  c(x_1,0)&c(x_2,0) & \ldots&c(x_r,0)\\
  \vdots &\vdots &\vdots & \vdots\\
  c(x_1,0)^{r-1}&c(x_2,0)^{r-1}&\ldots &c(x_r,0)^{r-1}.
 \end{array}\right|\ 
$$  
But $\Delta(x_1,\dots,x_r)$ is Vandermonde determinant, so
$$
\Delta(x_1,\dots,x_r)=\prod_{u<v}(c(x_u,0)-c(x_v,0))=\prod_{u<v}(j_u-j_v+k(i_u-i_v)),
$$
where the product is taken over all boxes $x_u \in \nu_u, x_v \in \nu_v.$
We can suppose that if $u<v$ then the connected component $\nu_u$ is located higher and more to the right than $\nu_v$, so
we have for $u<v$ that $j_u-j_v>0$ and $i_u-i_v<0$. Therefore if we consider $\Delta(x_1,\dots,x_r)$ as a polynomial in $k$  then its constant term is strictly negative and thus the same is true for $\Delta$. In the same way we can see that the coefficient at the highest degree of $k$ is strictly positive. Since $k$ is not algebraic number we see that $\Delta\ne0$.
\end{proof}

Let $\mathcal B^{(r)}$ be the CMS integrals (\ref{bil}) and consider the following system of linear equations
\begin{equation}
\label{array1}
\left\{ \begin{array}{r}
g_{1}(\nu_1)\mathcal M_1+g_{1}(\nu_2)\mathcal M_2+\dots+g_{1}(\nu_r)\mathcal M_r=\mathcal B^{(2)}-b_2 I\\
g_{2}(\nu_1)\mathcal M_1+g_{2}(\nu_2)\mathcal M_2+\dots+g_{2}(\nu_r)\mathcal M_r=\frac{1}{2}(\mathcal B^{(3)}-b_3 I)\\
\dots\dots\dots\dots\dots\dots\dots\dots\dots\dots\dots\dots\dots\dots\dots\ldots\\
g_{r}(\nu_1)\mathcal M_1+g_{r}(\nu_2)\mathcal M_2+\dots+g_{r}(\nu_r)\mathcal M_r=\frac{1}{r}(\mathcal B^{(r+1)}-b_{r+1} I)
\end{array}
\right.
\end{equation}
where the eigenvalue $b_s=b_s(\alpha, k, n+k^{-1}m)$ does not depend on $\alpha\in E$.
Since the determinant of this system is nonzero, the system has a unique solution $\mathcal M_1,\dots,\mathcal M_r$, which are certain CMS integrals. 
We claim that the image of $\mathcal M_i$ of under $\varphi$ give us required $\varepsilon_i.$

To show this consider the transition matrix $A=(a_{\beta\alpha}), \beta,\alpha\in E$ from the basis $P$ to $Q$ in $V_E(p_0):$
 $$
 Q_{\alpha}^{k,p_0}=\sum_{\beta\subset \alpha}a_{\beta\alpha}P_{\beta}^{k,p_0}.
 $$
 Let $A^{-1} =(\tilde a_{\beta\alpha})$ be the inverse matrix. It is easy to see that $\tilde a_{\beta\alpha}$ can be different from $0$ only  if $\beta\subset\alpha$. Now let $\nu$ be one of $\nu_1, \dots, \nu_r$ and define $2^r\times 2^r$ matrix
\begin{equation}
\label{varep}
\tilde\varepsilon_{\nu}=\sum_{\beta,\gamma,\alpha}\tilde a_{\beta\gamma} a_{\gamma\alpha}E_{\beta\alpha},
\end{equation}  
where the sum is taken over all triples $\beta\subset\gamma\subset\alpha$ from $E$ such that $\gamma\setminus\beta\supset\rho, \, \rho=(\nu, \theta(\nu))$ and $E_{\beta\alpha}, \, \alpha, \beta \in E$ are standard matrices
with only one non-zero matrix element $(\beta \alpha)$ equal to 1.

Let $D^{(s)}=D^{(s)}(p_0)$ be the matrix of the operator $\mathcal B^{(s)}$ in the basis $P_{\alpha}^{k,p_0}$, which is a diagonal matrix with the diagonal elements $d^{(s)}_{\alpha \alpha}=b_s(\alpha,k,p_0)$.
 Then the matrix of the operator $\mathcal B^{(s)}$ in the basis $Q_{\alpha}^{k,p_0}$ is $\tilde D^{(s)}=A^{-1}D^{(s)}A$.  
Consider the matrix 
 $$
 B^{(s)}=\tilde D^{(s)}-D^{(s)}=A^{-1}D^{(s)}A-D^{(s)}
  $$
with matrix elements
 $$
 b^{(s)}_{\beta\alpha}=\sum_{\beta\subset\gamma\subset\alpha}\tilde a_{\beta\gamma}d^{(s)}_{\gamma\gamma} a_{\gamma\alpha}-d^{(s)}_{\beta\alpha}=\sum_{\beta\subset\gamma\subset\alpha}\tilde a_{\beta\gamma}(d^{(s)}_{\gamma\gamma} -d^{(s)}_{\beta\beta})a_{\gamma\alpha},
 $$
where we have used that $A^{-1}A=I$. It is easy to see from the form of $d^{(s)}_{\beta \beta}=b_s(\beta,k,p_0)$ that 
 $$
 d^{(s)}_{\gamma\gamma}-d^{(s)}_{\beta\beta}=\sum_{\nu\subset\gamma\setminus\beta}\tilde g_s(\nu),
 $$
 where
 $$
 \tilde g_s(\nu)=\sum_{x\in\nu}c(x,0)^{s-1}+(-1)^s\sum_{x\in\theta(\nu)}c(x,1+k-kp_0)^{s-1}.
 $$
 Therefore the matrix $B^{(s)}$ can be represented in the form
   \begin{equation}
\label{bs}
 B^{(s)}=\sum_{\beta\subset\alpha}b^{(s)}_{\beta\alpha}E_{\beta\alpha}=\sum_{\nu}\tilde g_s(\nu)\tilde\varepsilon_{\nu}.
\end{equation}
 From this we see that the matrices $\tilde \varepsilon_{\nu}$ satisfy the following system of linear relations  
  \begin{equation}
\label{array2}
 \left\{ \begin{array}{r}
\tilde g_{2}(\nu_1)\tilde\varepsilon_1+\tilde g_{2}(\nu_2)\tilde\varepsilon_2+\dots+\tilde g_{2}(\nu_r)\tilde\varepsilon_r=\tilde D^{(2)}-D^{(2)}\\
\tilde g_{3}(\nu_1)\tilde\varepsilon_1+\tilde g_{3}(\nu_2)\tilde\varepsilon_2+\dots+\tilde g_{3}(\nu_r)\tilde\varepsilon_r=\tilde D^{(3)}-D^{(3)}\\
\dots\dots\dots\dots\dots\dots\dots\dots\dots\dots\dots\dots\dots\dots\dots\ldots\\
\tilde g_{r+1}(\nu_1)\tilde\varepsilon_1+\tilde g_{r+1}(\nu_2)\tilde\varepsilon_2+\dots+\tilde g_{r+1}(\nu_r)\tilde\varepsilon_r=\tilde D^{(r+1)}-D^{(r+1)}.
\end{array}
\right.
\end{equation}
 
 It will be convenient now to use instead of $p_0$ the local parameter
$$h=m+kn-kp_0,$$ such that $h=0$ when $p_0=n+k^{-1}m.$
 From the identity (\ref{ident}) we have 
 $$
 \tilde g_s(\nu)=\sum_{x\in\nu}c(x,0)^{s-1}-\sum_{x\in\nu}(c(x,0)-h)^{s-1},
 $$
 which implies that 
 $$
 \lim_{h\rightarrow 0}\frac{\tilde g_s}{h}=(s-1)g_{s-1}(\tau).
 $$
 From lemma \ref{lemm} the determinant
 $$
\tilde \Delta=\left|\begin{array}{cccc}
  \tilde g_2(\nu_1)/h&\tilde g_{2}(\nu_2)/h&\ldots &\tilde g_2(\nu_r)/h\\
  \tilde g_3(\nu_1)/h&\tilde g_{3}(\nu_2)/h & \ldots&\tilde g_3(\nu_r)/h\\
  \vdots &\vdots &\vdots & \vdots\\
  \tilde g_{r+1}(\nu_1)/h&\tilde g_{r+1}(\nu_2)/h&\ldots &\tilde g_{r+1}(\nu_r)/h
 \end{array}\right|\
$$
is not zero and since the right hand side is regular at $h=0$ we can define
\begin{equation}
\label{ep}
\varepsilon_{\nu}=\lim_{h\rightarrow 0}h\tilde\varepsilon_{\nu}.
\end{equation}

Taking limit $h \rightarrow 0$ in (\ref{array2}) and comparing the result with the system (\ref{array1}) we see that 
$\varepsilon_{i}=\varepsilon_{\nu_i}$ satisfy the same linear system as (and hence coincide with) $\varphi(\mathcal M_i)$ in the basis $Q_{\alpha}^{k,n+k^{-1}m}.$

Thus we have shown that $\varepsilon_{i}$ belong to the image of $\varphi.$
We claim now that $\varepsilon_i^2=0,\,i=1,\dots,r$ and that the products $\varepsilon_{i_1}\dots\varepsilon_{i_s}$ are linearly independent for all subsets $\{i_1,\dots,i_s\}\subset\{1,\dots,r\}$.
  
The relations  $\varepsilon_i^2=0$ follows from the equality $\tilde\varepsilon_i^2=0$, which is a simple consequence of the formula (\ref{varep}). Indeed, it is easy to see that for any two terms $E_{\beta\alpha}$ and $E_{\tilde\beta\tilde\alpha}$ entering (\ref{varep}) we have $\alpha\neq\tilde\beta$ since $\nu$ is a subset of $\alpha$, but not of $\tilde\beta.$

Define 
$
c^{i}_{\beta \alpha}=\sum \tilde a_{\beta\gamma} a_{\gamma\alpha},
$ 
where the sum is taken over all $\gamma \in E$ such that $\beta\subset\gamma\subset\alpha$ and $\gamma\setminus\beta\supset\rho_i, \, \rho_i=(\nu_i, \theta(\nu_i)).$ 
Then $$\tilde\varepsilon_{i}=\sum c^{i}_{\beta \alpha} E_{\beta \alpha},$$ where the sum is taken over $\alpha, \beta \in E$ such that $\beta \subset \alpha$ and $\alpha \setminus \beta$ contains $\rho_i.$

We have
 $$
\tilde\varepsilon_{i_1}\dots\tilde\varepsilon_{i_s}=\sum_{\beta\subset \alpha}c^{i_1,\dots,i_s}_{\beta,\alpha}E_{\beta,\alpha},
$$
where
$$
c^{i_1,\dots,i_r}_{\beta,\alpha}=\sum_{\beta\subset \gamma_1\subset\dots\subset \gamma_{s-1}\subset \alpha}c^{i_{1}}_{\beta \gamma_1}\dots c^{i_s}_{\gamma_{s-1}\alpha}
$$
and sum is taken for all possible chains such that $\rho_{i_1}\subset \gamma_1\setminus \beta$, $ \dots$, $\rho_{i_s}\subset \alpha\setminus \gamma_{s-1}$. 

If $\beta=\alpha_m$ is minimal in the sense of Theorem \ref{weight} and Corollary \ref{cor1} and $\alpha=\beta \cup \rho_{i_1} \cup \dots \cup  \rho_{i_s}$ then
 there is the only chain
$$
\beta\subset \beta \cup \rho_{i_1}\subset \beta \cup \rho_{i_1} \cup \rho_{i_2} \subset \dots\subset \alpha
$$
and 
$
c^{i_1,\dots,i_s}_{\beta,\alpha}=c^{i_{1}}_{\beta \gamma_1}\dots c^{i_s}_{\gamma_{s-1}\alpha}.
$
Now look at the coefficient $c^{i}_{\beta \alpha},$  where $\alpha = \beta \cup \rho_i.$ In that case
$
c^{i}_{\beta\alpha}=\tilde a_{\beta\alpha} a_{\alpha \alpha}= \tilde a_{\beta\alpha} = -a_{\beta\alpha}.$
From theorem \ref{regul} this coefficient has a pole of order 1, so the limit $hc^{i}_{\beta\alpha}$ when $h \rightarrow 0$ is non-zero.
Hence the product $\varepsilon_{i_1}\dots\varepsilon_{i_s}$ has a nonzero coefficient at $E_{\alpha_m \alpha_m \cup \rho_{i_1} \cup \dots \cup  \rho_{i_s}}.$ One can check that $E_{\alpha_m \alpha_m \cup \rho_{i_1} \cup \dots \cup  \rho_{i_s}}$ does not enter in any other product of $\varepsilon_i.$ This proves linear independence of $\varepsilon_{i_1}\dots\varepsilon_{i_s}.$ 

The fact that $\varepsilon_i, \, i=1,\dots, r$ generate the whole image of $\varphi$ follows from the formula (\ref{bs}) and from the fact that the operators $\mathcal B^{(l)}$ generate the algebra of CMS integrals.

Note that the commutativity of $\varepsilon_i$ (which follows from the commutativity of CMS integrals) imply some relations for the coefficients $a_{\beta\alpha}.$

To prove that the corresponding action of $\mathfrak A_r$ in $V_E$ is the regular representation consider the socle \footnote{We are grateful to Pavel Etingof for this idea.} of $\mathfrak A_r$ generated by the product $S=\varepsilon_1\dots \varepsilon_r \in \mathfrak A_r.$ The action of $\mathfrak A_r$ in $V_E$ is faithful, so there is a vector $v \in V_E$ such that $Sv\neq 0.$ Since $S$ belongs to all non-zero ideals of $\mathfrak A_r$ the subspace $\mathfrak A_r v \subseteq V_E$ is the regular representation of $\mathfrak A_r.$ Now the claim follows since  $\mathfrak A_r v$ and $V_E$ have the same dimension $2^r$ and thus must coincide.
\end{proof}

\section{Concluding remarks}

The behaviour of Jack symmetric functions for special (namely, positive rational) values of parameter $k$  are known to be quite tricky and is still to be properly understood.  As it was shown by B. Feigin et al \cite{FJMM}  this question turned out to be closely related with the classical  coincident root loci problem going back to Sylvester and Cayley (see \cite{KMSV}).

We have shown that the Jack--Laurent case turns out to be much simpler in this respect
and the analytic properties of the coefficients can be described in a satisfactory manner (see section 3 above). 
The reason is that in this case we have two parameters $k$ and $p_0,$ 
and we can fix $k$ to be generic and consider the analytic properties in $p_0$ instead.

Our main motivation to study Jack-Laurent symmetric functions came from the representation theory of Lie superalgebras, where the case of special parameters is particularly important. We will discuss this in a separate publication.
  
\section{Acknowledgements}

We are grateful to P. Etingof for very helpful remarks and to B. Feigin and J. Shiraishi for stimulating discussions.

This work was partly supported by the EPSRC (grant EP/J00488X/1). ANS is grateful to Loughborough University for the hospitality during the autumn semesters 2012-14.

\end{document}